\documentclass[rmp,onecolumn]{revtex4-2}
\usepackage{bm}
\usepackage{graphicx}
\usepackage{verbatim}
\usepackage{amsthm}
\usepackage{float}
\usepackage[section]{placeins}
\usepackage{amsmath,amssymb}
\newtheorem{lemma}{Lemma}
\usepackage{hyperref}

\begin{document}
\title{Probing frustrated spin systems with impurities%
}
\author{Maksymilian Kliczkowski}
\affiliation{Institute of Theoretical Physics, Wroc{\l}aw University of Science and Technology, 50-370 Wroc{\l}aw, Poland}
\author{Jakub Grabowski}
\affiliation{Institute of Theoretical Physics, Wroc{\l}aw University of Science and Technology, 50-370 Wroc{\l}aw, Poland}
\author{Maciej M. Ma\'ska}
\email[corresponding author: ]{maciej.maska@pwr.edu.pl}
\affiliation{Institute of Theoretical Physics, Wroc{\l}aw University of Science and Technology, 50-370 Wroc{\l}aw, Poland}

\begin{abstract}
We investigate the effective interaction between two localized spin impurities embedded in a frustrated spin--$\frac{1}{2}$ $J_1\!-\!J_2$ Heisenberg chain. Treating the impurity spins as classical moments coupled locally to the host, we combine second-order perturbation theory with large-scale density matrix renormalization group (DMRG) calculations to determine the impurity–impurity interaction as a function of separation, coupling strength, and magnetic frustration. In the weak-coupling regime, we show that the interaction is governed by the the static spin susceptibility of the host and exhibits oscillatory power-law decay in the gapless phase, modified by universal logarithmic corrections at the SU(2)--symmetric critical point. In the gapped dimerized phase, the interaction decays exponentially with distance. For intermediate and strong impurity–host coupling, we observe a crossover to a boundary-dominated regime characterized by pronounced parity effects associated with the length of the chain segment between impurities, signaling a breakdown of the simple RKKY--like description. Our results establish impurity--impurity interactions as a sensitive probe of frustrated quantum spin liquids and provide a controlled framework for distinguishing gapless and gapped phases through local perturbations.
\end{abstract}

\maketitle
  
\section{Introduction}
Quantum spin liquids (QSLs) are paradigmatic examples of strongly correlated quantum matter, characterized by the absence of conventional symmetry-breaking magnetic order even at zero temperature and by the presence of highly entangled ground states and fractionalized excitations \cite{wen_quantumorders_2002, Balents2010,Savary2017}. While much of the modern interest in QSLs focuses on two-dimensional frustrated magnets~\cite{sachdev_kagometriangularlattice_1992,balents_fractionalizationeasyaxis_2002, senthil_microscopicmodels_2002, motrunich_exoticorder_2002, jiang_densitymatrix_2008, yan_spinliquidground_2011, schmidt_frustratedtwo_2017, broholm_quantumspin_2020, zhu_quantumspin_2025, matsuda_kitaevquantum_2025a}, one-dimensional quantum spin systems provide an essential theoretical laboratory in which many key features of spin-liquid physics--such as strong quantum fluctuations, the absence of long-range order, and spin fractionalization--can be studied in a controlled and often analytically tractable setting~\cite{dasgupta_lowtemperatureproperties_1980,fisher_randomantiferromagnetic_1994,sato_competingphases_2011,imambekov_onedimensionalquantum_2012,lecheminant_onedimensionalquantum_2003, amusia_onedimensionalquantum_2020}.

A canonical example is the spin--$\frac{1}{2}$ frustrated $J_1–J_2$ Heisenberg chain. By tuning the ratio $J_2/J_1$, this model exhibits a quantum phase transition from a gapless Luttinger-liquid phase with algebraically decaying spin correlations to a gapped, spontaneously dimerized phase with short-range correlations \cite{Haldane1982,Majumdar1969,Okamoto1992}. The gapless phase is often regarded as a one-dimensional quantum spin liquid in a broad sense, featuring deconfined spin--$\frac{1}{2}$ spinon excitations and power-law correlations, while the gapped phase provides a minimal realization of a valence-bond-solid state induced by frustration. Throughout this work, we use the term ``quantum spin liquid'' to refer to a magnetically disordered state with fractionalized excitations and algebraic correlations. This definition does not imply topological order.

Magnetic impurities constitute a powerful probe of correlated quantum spin systems~\cite{eggert_phasediagram_2001, laflorencie_kondoeffect_2008, bayat_entanglementprobe_2012}. A single impurity locally perturbs the host and generates an extended polarization cloud that reflects the intrinsic spin correlations and low-energy excitations of the system \cite{Eggert1992,Laflorencie2006}. When two or more impurities are present, their interaction is mediated by the host degrees of freedom and encodes detailed information about the static and dynamical spin response. In itinerant electron systems this mechanism is well known as the Ruderman–Kittel–Kasuya–Yosida (RKKY) interaction \cite{Ruderman1954,Kasuya1956,Yosida1957, bayat_entanglementprobe_2012}. In quantum spin systems, however, the absence of fermionic quasiparticles and the presence of strong correlations lead to qualitatively distinct behavior, raising fundamental questions regarding the spatial range, oscillatory structure, and symmetry properties of impurity-induced interactions. As an example, a characteristic critical behavior of Friedel oscillations was recently reported in the spin--$\frac{1}{2}$ frustrated $J_1–J_2$ Heisenberg chain in the spin-nematic regime, in the presence of impurities of the local-bond disorder type~\cite{nakamura_impurityeffects_2024}.

In one-dimensional spin chains, the impurity problem has attracted long-standing interest, particularly in the context of Kondo physics and boundary critical phenomena \cite{Affleck1991, laflorencie_kondoeffect_2008}. In gapless spin chains, impurity-induced correlations decay algebraically, suggesting long-ranged effective interactions between impurities, while in gapped systems such interactions are expected to be exponentially suppressed beyond the correlation length. In frustrated chains, additional richness arises from incommensurate correlations and frustration-induced oscillations, which can generate competing effective couplings and potentially favor non-collinear or helical arrangements of impurity moments, even though the host system itself remains non-magnetic~\cite{eggert_phasediagram_2001, sarkar_quantumphase_2008}.

Understanding impurity-impurity interactions in frustrated spin chains is therefore closely connected to quantum spin liquids~\cite{kolezhuk_theoryquantum_2006}, where Kondo physics with highly unconventional characteristics may arise~\cite{dhochak_magneticimpurities_2010, das_kondoroute_2016, lee_kondoscreening_2023, lu_detectingsymmetry_2024, legg_spinliquid_2019, sorensen_kondoscreening_2024}. First, such interactions provide an indirect, yet sensitive diagnostic of the host's excitation spectrum and correlation length, allowing one to distinguish gapless and gapped quantum phases without directly probing bulk excitations. Second, impurity spins may develop collective behavior--such as frustration-driven non-collinear order--that reflects the underlying spin-liquid correlations of the host. Finally, impurity physics is of practical relevance to real materials, where disorder and substitutional defects are unavoidable, and to engineered quantum systems, including cold-atom simulators and solid-state spin-chain platforms.

Previous studies of impurities in one-dimensional quantum spin chains have focused primarily on single-impurity physics~\cite{laflorencie_kondoeffect_2008}, boundary critical phenomena, and Kondo-like screening effects in gapless systems \cite{Eggert1992,Affleck1991,Laflorencie2006}. These works have provided a detailed understanding of how a localized defect modifies local observables, induces polarization clouds, or renormalizes boundary conditions at low energies. By contrast, the problem of effective interactions between multiple impurities--particularly in frustrated spin chains--has received comparatively less attention~\cite{bayat_entanglementprobe_2012}. In such settings, impurities are not merely passive probes of the host but can interact non-trivially through the collective spin degrees of freedom, giving rise to distance-dependent, oscillatory, and potentially frustrated effective couplings.

In the present work, we move beyond the single-impurity perspective and explicitly focus on the two-impurity problem in the frustrated $J_1\!-\!J_2$ Heisenberg chain. This allows us to directly connect impurity-impurity interactions to the intrinsic correlation functions of the host and to elucidate how these interactions evolve across the transition between gapless and gapped quantum phases. Unlike studies centered on boundary impurities or Kondo screening, we treat impurities as local moments embedded in the bulk and analyze their mutual interaction as a function of separation and orientation. This approach highlights the role of frustration and incommensurate correlations in shaping impurity-induced physics and establishes impurity interactions as a sensitive probe of one-dimensional quantum spin liquids. To our knowledge, a systematic comparison between perturbative susceptibility-based predictions and nonperturbative DMRG calculations of impurity–impurity interactions across the $J_1\!-\!J_2$ phase diagram has not been reported.

\section{Model}

We consider a frustrated spin--$\frac{1}{2}$ Heisenberg chain locally coupled to two impurity spins. 
The Hamiltonian is
\begin{equation}
    H = H_{J_1\!-\!J_2} + H_{\rm imp},
\end{equation}
where
\begin{align}
    H_{J_1\!-\!J_2} &= J_1\sum_i \hat{\bm S}_i\cdot \hat{\bm S}_{i+1}
    + J_2\sum_i \hat{\bm S}_i\cdot \hat{\bm S}_{i+2},\\
    H_{\rm imp} &= J_c\left({\bm S}_{c,1}\cdot\hat{\bm S}_i
    + {\bm S}_{c,2}\cdot\hat{\bm S}_j\right).
\end{align}
Here, $\hat{\bm S}_i$ denotes a quantum spin--$\frac{1}{2}$ operator at the chain site $i$, while 
${\bm S}_{c,1}$ and ${\bm S}_{c,2}$ are classical impurity spins coupled to chain sites $i$ and $j$, respectively.
Our goal is to derive an effective interaction between the impurity spins mediated by the spin chain,
\begin{equation}
    V({\bm S}_{c,1},{\bm S}_{c,2},r) = \sum_{a,b} J_{ab}(r)\, S_{c,1}^a S_{c,2}^b,
    \label{eq:H_eff}
\end{equation}
where $r=|i-j|$ is the separation of impurities. We are particularly interested in how the decay of the interaction with increasing $r$ depends on whether the bulk system's spectrum is gapped or gapless. This may be useful for probing the nature of the underlying quantum state using magnetic impurities in QSLs.

\section{Weak-coupling regime}

We first focus on the weak-coupling regime $J_c \ll J_1, J_2$, where the impurity–chain interaction can be treated perturbatively.
Using standard Rayleigh--Schrödinger perturbation theory, the second-order correction to the ground-state energy reads
\begin{subequations}
\begin{align}
    \Delta E^{(2)}
    =&\:-\sum_{n\neq 0}\frac{\langle 0| H_{\rm imp}| n\rangle\langle n|H_{\rm imp}|0\rangle}{E_n-E_0}\\
    =&\: -J_c^2 \sum_{a} \left(S_{c,1}^a\right)^2 \sum_{n\neq 0}
    \frac{\langle 0| \hat{S}^a_i| n\rangle\langle n|\hat{S}^a_i|0\rangle 
    + \langle 0| \hat{S}^a_i| n\rangle\langle n|\hat{S}^a_i|0\rangle}{E_n-E_0}\\
    &\:-J_c^2\sum_{b}\left(S_{c,2}^b\right)^2\sum_{n\neq 0}\frac{
    \langle 0| \hat{S}^b_j| n\rangle\langle n|\hat{S}^b_j|0\rangle
    +\langle 0| \hat{S}^b_j| n\rangle\langle n|\hat{S}^b_j|0\rangle}{E_n-E_0}\\
    &\:-J_c^2\sum_{a,b}S_{c,1}^a S_{c,2}^b\sum_{n\neq 0}
    \frac{\langle 0| \hat{S}^a_i| n\rangle\langle n|\hat{S}^b_j|0\rangle
    +\langle 0| \hat{S}^b_j| n\rangle\langle n|\hat{S}^a_i|0\rangle}{E_n-E_0}\label{eq:cross-term}\\
    =&\:\Delta E^{(2)}_{11}+\Delta E^{(2)}_{22}+\Delta E^{(2)}_{12}
\label{eq:sec_ord}
\end{align}
\end{subequations}
where $\{|n\rangle\}$ are the excited eigenstates of $H_{J_1\!-\!J_2}$ with energies $E_n$. $\:\Delta E^{(2)}_{11}$ ($\Delta E^{(2)}_{22}$) corresponds to impurity 1 (2) interacting with itself via the chain. $\Delta E^{(2)}_{12}$ is the cross term that describes the effective interaction between impurities 1 and 2. For now, we will focus on this term.

The sum over intermediate states in Eq.~\eqref{eq:cross-term} can be related to the retarded spin susceptibility of the host chain,
\begin{equation}
    \chi_{ab}(r,t)
    =
    -i\,\Theta(t)
    \left\langle
    \left[
    \hat{S}^a_i(t),
    \hat{S}^b_j(0)
    \right]
    \right\rangle,
\end{equation}
where $r=|i-j|$ and
$\hat{S}^a_i(t)=e^{iH_{J_1J_2}t} \hat{S}^a_i e^{-iH_{J_1J_2}t}$.
Its Fourier transform,
\begin{equation}
    \chi_{ab}(r,\omega)
    =
    \int_0^\infty dt\,
    e^{i(\omega+i0^+)t}
    (-i)
    \left\langle
    \left[
    \hat{S}^a_i(t),
    \hat{S}^b_j(0)
    \right]
    \right\rangle,
\end{equation}
admits the Lehmann representation
\begin{equation}
\chi_{ab}(r,\omega) = \sum_{n\neq 0} \left( \frac{
\langle 0| \hat{S}^a_i| n\rangle\langle n|\hat{S}^b_j|0\rangle
}{\omega-(E_n-E_0)+i0^+} - \frac{\langle 0| \hat{S}^b_j| n\rangle\langle n|\hat{S}^a_i|0\rangle
}{\omega+(E_n-E_0)+i0^+}\right).
\label{eq:lehmann}
\end{equation}
Taking the real part of Eq.~\eqref{eq:lehmann} at $\omega=0$, one finds that it coincides with the sum appearing in Eq.~\eqref{eq:sec_ord}.
As a result, the impurity--impurity interaction energy can be written as
\begin{equation}
    V({\bm S}_{c,1},{\bm S}_{c,2},r) = -J_c^2 \sum_{a,b} S_{c,1}^a S_{c,2}^b \,{\rm Re}\,\chi_{ab}(r,\omega=0).
    \label{eq:final1}
\end{equation}
Equation~\eqref{eq:final1} is the direct analog of the RKKY interaction in itinerant electron systems, with the electronic spin susceptibility replaced by the static spin susceptibility of the quantum spin chain.
Thus, in the weak-coupling regime, the effective interaction between impurities is entirely governed by the static spin response of the host.

If the spin chain is SU(2)-symmetric and invariant under time reversal, the susceptibility tensor is diagonal,
$\chi_{ab}(r,\omega)=\delta_{ab}\,\chi(r,\omega)$, and Eq.~\eqref{eq:final1} reduces to the isotropic bilinear form
\begin{equation}
    V({\bm S}_{c,1},{\bm S}_{c,2},r) = -J_c^2\, {\rm Re}\,\chi(r,\omega=0)\,{\bm S}_{c,1}\cdot{\bm S}_{c,2}.
\end{equation}
Since the magnitude of the classical spins is fixed $S_c=|{\bm S}_{c,i}|$, $H_{\rm eff}$ depends only on $r$ and the angle $\theta$ between ${\bm S}_{c,1}$ and ${\bm S}_{c,1}$,
\begin{equation}
    V(r,\theta) = -(J_c\,S_c)^2\, {\rm Re}\,\chi(r,\omega=0)\, \cos\theta.
    \label{eq:H_eff_perturb}
\end{equation}

Although the full lattice-scale form of ${\rm Re}\,\chi(r,\omega=0)$ is not known analytically for the $J_1$--$J_2$ Heisenberg chain, its long-distance behavior is well established. In particular, the low-energy limit of this model is well
understood in terms of the Wess-Zumino-Witten nonlinear $\sigma$ model with $k=1$ \cite{Affleck1990}. For $r\gg 1$ one finds
\begin{itemize}
    \item
    ${\rm Re}\,\chi(r,\omega=0)\sim (-1)^r\,r^{-1}$ in the gapless phase,
    $J_2/J_1 < J_2^c \approx 0.2411$, as follows from the Luttinger-liquid
    description of the spin--$\frac{1}{2}$ Heisenberg chain
    \cite{Okamoto1992,Sandvik2010,Kumar2007}.
    \item
    At the SU(2)-symmetric point, this power-law decay is modified by an universal
    multiplicative logarithmic corrections originating from a marginally
    irrelevant current--current interaction in the low-energy theory,
    \begin{equation}
        {\rm Re}\,\chi(r,\omega=0)
        \sim
        (-1)^r\,\frac{\sqrt{\ln(r/r_0)}}{r},
    \end{equation}
    where $r_0$ is a nonuniversal microscopic length scale
    \cite{Eggert1996}.
    \item
    ${\rm Re}\,\chi(r,\omega=0)\sim (-1)^r\,r^{-1/2}e^{-r/\xi}$, where $\xi$ denotes the spin correlation length, in the gapped, dimerized phase, $J_2/J_1 > J_2^c$, as described by the massive sine-Gordon field theory \cite{Haldane1982,Okamoto1992}. The exponent $1/2$ in this phase reflects the relativistic massive field-theory description of the dimerized state.
    \item 
    In this phase, there is a special point, $J_2=J_1/2$, known as the Majumdar-Ghosh (MG) limit \cite{MG69,MG69a,Caspers84,Ravi07}. At this point, the ground state is an exactly dimerized valence-bond state with strictly short-ranged spin correlations. Consequently, in the weak-coupling regime, the impurity–impurity interaction is exponentially short-ranged, with a correlation length of the order of one lattice spacing.
\end{itemize}

\section{Strong-coupling regime}\label{sec:strong-coupling}
When $J_c\gg J_1,\,J_2$, each impurity no longer acts as a weak probe of the spin liquid. Instead, it pins a quantum spin into a strong composite object with the impurity. For $J_c\to\infty$, the quantum spin at the impurity site is no longer a dynamical degree of freedom. Rather, it becomes a static boundary condition for the remaining chain. As a result, the chain is effectively cut into three segments, as illustrated in Fig.~\ref{fig:3sections}.
\begin{figure}[!h]
    \centering
    \includegraphics[width=0.6\linewidth]{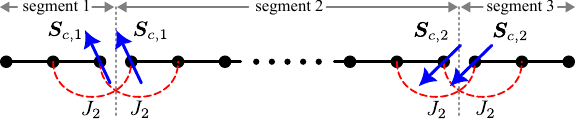}
    \caption{Illustration how the chain is divided into three segments for $J_c\to\infty$.}
    \label{fig:3sections}
\end{figure}
In addition to having shared boundary conditions, the individual segments are coupled via $J_2$, as shown by the red dashed lines. While the polarization of the spin at the impurity-coupled site is fixed, the rest of the system remains in the form of quantum chains whose lengths depend on the positions of the impurities. In particular, if the impurities are coupled to lattice sites $i_1$ and $i_2$, the length of the chain between them is $l = i_2 - i_1 - 1$. For even $l$, the chain has a singlet ground state, whereas for odd $l$, it is an effective spin $1/2$ doublet. Therefore, unlike the weak-coupling regime, the total ground-state energy depends sensitively on the parity of the segment length between impurities. This leads to an intrinsic alternation of the two-impurity energy with separation. This non-perturbative parity effect cannot be absorbed into a smooth envelope, which signals the breakdown of the weak-coupling RKKY description. In the appendix, we prove that the fact that the orientation of classical spins enters the energy only through the dot product of the vectors is a general property of SU(2)-symmetric systems and is not limited to the weak-coupling regime. However, unlike in the weak-coupling case, the dependence of the energy on the dot product can be nonlinear. This implies that the interaction energy is not simply proportional to $\cos\theta$.

Unlike in the weak-coupling regime, the MG limit can be nontrivial for $J_c \gtrsim J_1, J_2$. A strong impurity at a given lattice site breaks a local singlet, creating a free spin--$\frac{1}{2}$ soliton that can propagate along the chain at a finite energy cost. In the case of two impurities, the two released spins interact, and their interaction energy depends on their relative position. This produces an impurity-impurity interaction mediated by spinon exchange rather than by conventional susceptibility, as in the weak-coupling regime. The effective interaction decays exponentially with a decay length set by the spinon localization length. However, this interesting regime is beyond the scope of this work.

For $J_c\to\infty$ the quantum spins $\hat{\bm S}_i$ at the impurity sites are fully polarized along the classical spins ${\bm S}_{c,1/2}$. Therefore, there are no low-energy fluctuations left at these sites. This limit is exactly solvable, because it reduces to open Heisenberg chains plus constants. For large but finite $J_c$, one can perform a Schrieffer–Wolff (strong-coupling) expansion in powers of $J_{1,2}/J_c$. To leading order, corrections scale as $J_{1,2}^2/J_c$, indicating that the strong-coupling fixed point is stable. However, this is beyond the scope of this work.

The transition between the weak- and strong-coupling regimes as a function of $J_c$ constitutes a crossover between bulk-mediated and boundary-mediated impurity interactions.

\section{General case}
When the value of the coupling $J_c$ between the impurities and the $J_1–J_2$ chain does not allow for a perturbative treatment, numerical methods must be applied to determine the effective impurity-impurity interaction. For this purpose, we apply the Density Matrix Renormalization Group (DMRG) approach. This procedure involves finding the ground state energy as a function of the impurity spins' orientation and the distance between the impurities. Therefore,  Eq.~\eqref{eq:H_eff_perturb} needs to be generalized,
\begin{equation}
    V(r,\theta) = -(J_c\,S_c)^2\,f(r,\cos\theta),
    \label{eq:Jc_scaling}
\end{equation}
where $f(r,\cos\theta)$ must be determined numerically. 
\section{Numerical results}
To numerically calculate the effective interaction potential between the impurity spins, we calculate
\begin{align}
    V(r,\theta) &= \left(E_{12}(r,\cos\theta) - E_0\right) - (E_1 - E_0) - (E_2 - E_0) \nonumber \\
    &= E_{12}(r,\cos\theta) - E_1 - E_2 + E_0,
    \label{eq:int_en}
\end{align}
where $E_{12}(r,\cos\theta)$, $E_{1,2}$, and $E_0$, are the ground state energies with two, one, and no impurities, respectively. In the thermodynamic limit, for fixed $J_cS_c$, $E_{12}(r,\cos\theta)$ depends solely on the relative separation of the impurities and on the angle between their spins, while $E_1=E_2$ and is independent of both the impurity’s position and the orientation of its spin. However, in DMRG calculations, a finite chain with open boundaries is used and in this case the values of both $E_{12}(r,\cos\theta)$ and $E_{1,2}$ depend on the absolute positions of the impurities. This is particularly pronounced when $J_2/J_1$ is close to the critical value, where boundary effects decay only algebraically. To minimize this effect, we position impurities 1 and 2 symmetrically around the center of the chain. This has two advantages: ({\em i\,}) for small and intermediate distances $r$, the impurities are relatively far from the boundaries of the chain, and ({\em ii\,}) the boundary-induced Friedel oscillations cancel out as much as possible. 
\begin{figure}[!htb]
    \centering
    \includegraphics[width=0.52\linewidth]{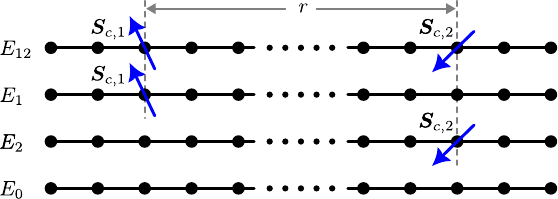}
    \caption{Illustration how the components of the interaction energy \eqref{eq:int_en} is calculated.}
    \label{fig:int_en}
\end{figure}
Additionally, we calculate $E_1$ and $E_2$ for impurities 1 and 2 located on the same lattice sites as when calculating $E_{12}(r,\cos\theta)$, as illustrated in Fig.~\ref{fig:int_en}.

\begin{figure}[!htbp]
    \centering
    \includegraphics[width=0.8\linewidth]{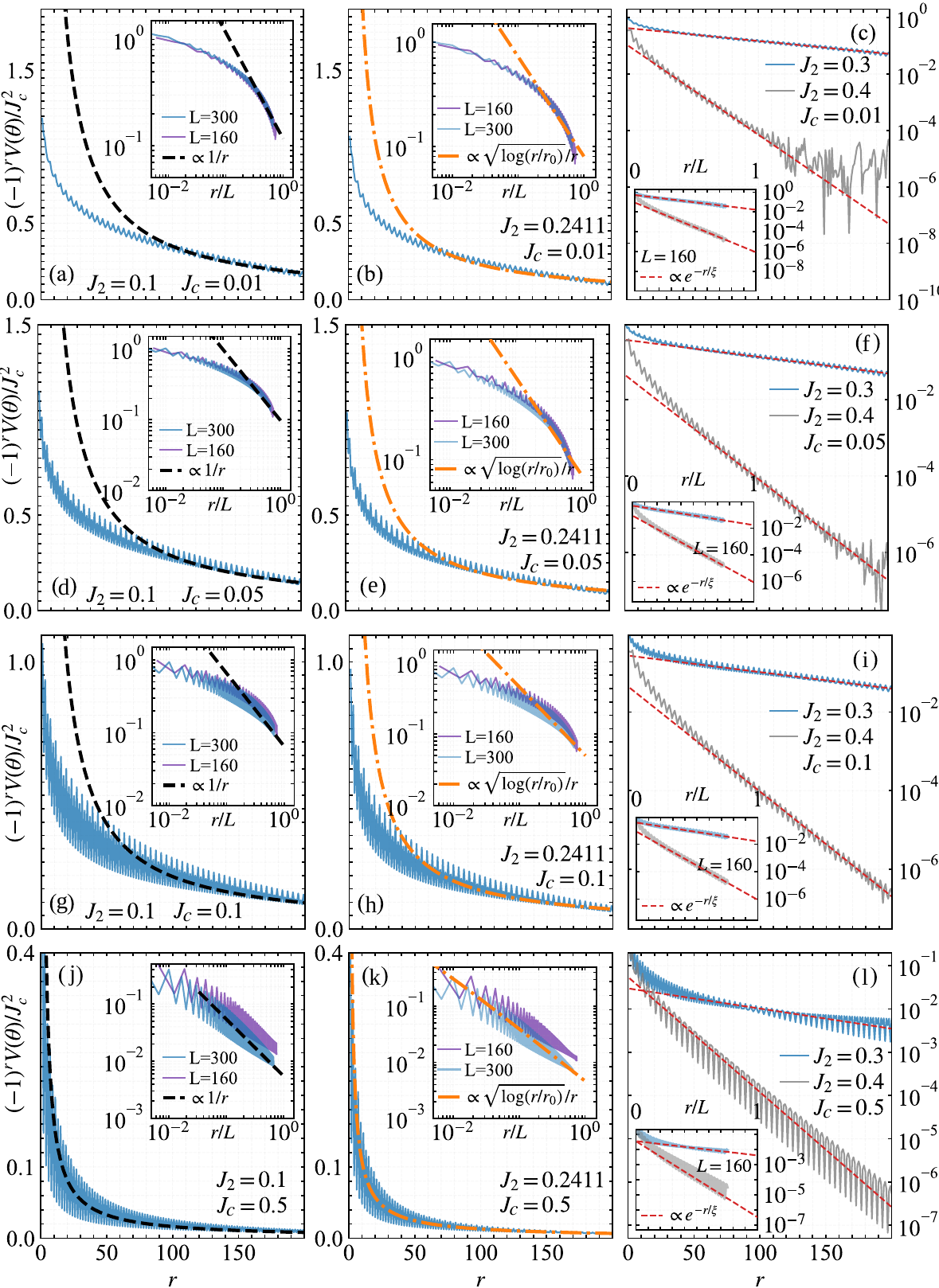}
    \caption{The effective interaction, $V(r,\theta=0)$ (cf. Eq.~\eqref{eq:int_en}), for different values of $J_2$ and $J_c$. The first and second columns show the results for $J_2=0.1$ and $J_2=J_2^c$, respectively. The third column shows the results for $J_2=0.3$ and $J_2=0.4$. The first, second, third, and fourth rows show the results for $J_c$ equal to $0.01, 0.05, 0.1$, and $0.5$, respectively.}
    \label{fig:phase_diagrams}
\end{figure}

The results are summarized in Fig.~\ref{fig:phase_diagrams}, where the effective interaction is shown for different $J_2$ (between 0.1 and 0.4) and different $J_c$ (between 0.1 and 0.5). Two system sizes were studied: $L=160$ and $L=300$ lattice sites. Since the perturbation approach predicts oscillating behavior in all cases, with the sign of the effective interaction changing every lattice site, the results presented are multiplied by $(-1)^r$, where $r$ is the distance between the impurities. Furthermore, according to Eq.~\eqref{eq:Jc_scaling}, the results were normalized by dividing by $J_c$.

On the basis of these plots, several observations can be made. Since the perturbation results are for $r\to\infty$, it cannot be expected that they fit the numerical curves precisely at small distances. Disagreement is observed mainly for small $J_2$ and small $J_c$. Interestingly, the power-law perturbative formula works for $J_c=0.5$, which is when the perturbative approach is least valid. However, while the exponential decay of the effective interaction can be clearly visible for $J_2>J_2^c$ and almost all values of $J_c$, for $J_c\lesssim J_2^c$ the numerical results deviate from the perturbative predictions even when the impurities are distant. This is particularly evident in the insets of panels (a),(b),(d),(e),(g), and (h), which present the same results on a log-log scale. In this case, power-law behavior should manifest itself as straight lines. We attribute these deviations mainly to finite size effects. First, the deviations are mostly present for $J_2 \lesssim J_2^c$, where the correlation length is clearly much larger than for $J_2 > J_2^c$, for which the exponential decay fits very well. Furthermore, these deviations decrease as $J_c$ increases, i.e., as the range of the effective interaction decreases.

Another observation is that the magnitude of the effective interaction oscillates. A detailed analysis reveals that the main period is four lattice constants. This arises from the superposition of oscillations with two slightly different periods, each corresponding to twice the lattice constant. These oscillations clearly increase with increasing $J_c$. The occurrence of these oscillations indicates that the system is no longer in the RKKY-like state, but is approaching the strong coupling regime. In this regime, the energy depends on the parity of the section between the impurities, as described in Sec.~\ref{sec:strong-coupling}. The complex pattern of these oscillations results from the fact that energy depends not only on the length of the section between the impurities, but also on the lengths of the sections between the impurities and the edges of the chain. This can be inferred from the dependence of energy on the position of a single impurity ($E_1, E_2$, cf. Eq.~\eqref{eq:int_en}) coupled to the chain (not shown).

We also performed DMRG calculations for larger values of $J_2$ and $J_c$, but the results of these calculations are not shown in the plots. They can be summarized as follows: as $J_2$ approaches the MG point $J_1/2$, the correlations become extremely short. At precisely $J_2=J_1/2$, the ground state is an exact product of nearest-neighbor singlets, resulting in the correlations limited to one lattice constant. Beyond the MG point, for $J_2$ slightly larger than $J_1/2$, the dimerized phase persists but the ground state acquires quantum fluctuations, leading to a finite correlation length and a short-ranged but nonvanishing impurity interaction. 

For stronger coupling, the results generally follow the tendency already observed for $J_c$ increasing from 0.01 to 0.5: the energy exhibits an increasingly pronounced even-odd dependence on the length of the segment between impurities. Consequently, $V(r)$ exhibits an intrinsic alternation that cannot be absorbed into a smooth envelope by multiplying by $(-1)^r$. 

The results are presented only for the angle between the impurity spins $\theta=0$. This is sufficient in the weak-coupling regime, where the interaction in an arbitrary configuration can be obtained by multiplying the effective potential by $\cos\theta$. While the SU(2) symmetry ensures that the interaction depends only on $\cos\theta$ also beyond this regime, for larger $J_c$, the dependence can be nonlinear. However, verifying the full angular dependence numerically in the strong-coupling regime remains computationally demanding and is left for future work

\section{Summary and Outlook}
In this work, we have investigated the effective interaction between two localized spin impurities coupled to a frustrated one-dimensional $J_1\!-\!J_2$ Heisenberg chain, combining controlled perturbation theory with large-scale DMRG calculations. By treating the impurity spins as classical vectors of fixed magnitude, we were able to extract the impurity–impurity interaction energy directly from ground-state energies and to analyze its dependence on impurity separation, coupling strength, and relative orientation.

In the weak-coupling regime $J_c\ll J_1,J_2$, we demonstrated that the effective interaction is governed by the static spin susceptibility of the host, in close analogy with the RKKY mechanism. The distance dependence of the interaction thus provides direct access to the long-distance spin correlations of the underlying spin liquid. In particular, the oscillatory behavior and decay of the interaction encode whether the host is in a gapless or gapped quantum spin liquid phase. At the SU(2)-symmetric critical point that separates these phases, our numerical results are consistent, within accessible system sizes, with the field-theoretical prediction of multiplicative logarithmic corrections to power-law decay, highlighting the sensitivity of impurity probes to subtle marginal effects. 

These perturbative results have been confirmed by DMRG calculations, which extend beyond the weak-coupling regime to provide results also for intermediate and strong coupling $J_c$. In particular, we observed that for $J_c$ comparable to $J_1,J_2$, the impurity-impurity interaction is strongly affected by the parity effects associated with the length of the intervening segment and cannot be described by a simple RKKY picture. This is consistent with the scenario where the impurities pin local host spins and effectively impose boundary conditions that cut the chain into finite segments. This regime reveals a crossover from a bulk-response-dominated interaction to a boundary-dominated one, providing a clear diagnostic of the breakdown of linear-response descriptions.

More broadly, our results establish impurity–impurity interactions as a versatile probe of quantum spin liquids. Unlike conventional correlation functions, impurity interactions integrate dynamical spin correlations over space and time and can therefore amplify long-distance and low-energy features that are otherwise difficult to resolve. From an experimental perspective, this suggests that controlled magnetic defects, adatoms, or substituted spins could be used to infer the nature of spin-liquid correlations indirectly through their effective interactions. In solid-state systems, magnetic or nonmagnetic defects-introduced either intrinsically, by chemical substitution, or controllably, for example via adatoms in scanning tunneling microscopy-act as localized perturbations that couple to the surrounding spin environment. The effective interaction between such impurities, inferred from local spectroscopy or thermodynamic signatures, offers indirect but sensitive access to the host’s spin correlations without requiring direct measurement of bulk dynamical response functions. Complementary opportunities arise in cold-atom and programmable quantum simulator platforms, where localized spins or controlled defects can be engineered and positioned with high precision. In these settings, impurity–impurity interactions could serve as tunable probes of correlation length, excitation gaps, and critical behavior, providing a versatile experimental diagnostic of quantum spin liquid phases across different dimensions and microscopic realizations.

Several directions for future work naturally emerge. Extending this approach to two-dimensional quantum spin liquids--such as Kitaev, kagome, or triangular-lattice systems--would allow one to probe fractionalization, gauge-field-mediated interactions, and topological effects through impurity responses. Generalizing the impurities to quantum spins would enable the study of Kondo-like screening and impurity entanglement in spin liquids. Finally, exploring the effects of anisotropies, external fields, or nonequilibrium driving could further enrich the use of impurities as diagnostic tools for exotic quantum phases.

\begin{acknowledgments}
The authors thank Nandini Trivedi for valuable discussions. M.K. would like to thank her for the hospitality during his stay at athe Ohio State University under the Bekker Programme of the Polish National Agency for Academic Exchange Grant no. BPN/BEK/2024/1/00115. This work was supported by the National Science Centre (Poland) under Grant No. 2024/53/B/ST3/02756. Numerical calculations were carried out using the resources provided by the Wrocław Centre for Networking and Supercomputing.
\end{acknowledgments}
\appendix
\section{Symmetry constraint on the impurity--impurity interaction}

\begin{lemma}[Rotational invariance of the impurity energy]
Consider a spin system described by the Hamiltonian
\begin{equation}
H(\bm S_{c,1},\bm S_{c,2})
=
H_{\rm host}
+
J_c\left(
\bm S_{c,1}\cdot \hat{\bm S}_i
+
\bm S_{c,2}\cdot \hat{\bm S}_j
\right),
\end{equation}
where:
\begin{enumerate}
    \item $H_{\rm host}$ is invariant under global SU(2) spin rotations,
    \item the impurity--host couplings are isotropic (Heisenberg form),
    \item $\bm S_{c,1}$ and $\bm S_{c,2}$ are 
    classical vectors of fixed length.
\end{enumerate}
Then the spectrum of $H$, and in particular its ground-state energy, is invariant under a simultaneous rotation of the impurity spins,
\begin{equation}
E(\bm S_{c,1},\bm S_{c,2})
=
E(R\bm S_{c,1},R\bm S_{c,2})
\qquad \forall\,R\in{\rm SO}(3).
\end{equation}
As a consequence, for fixed impurity spin magnitudes the ground-state energy can depend on the impurity orientations only through rotationally invariant combinations. In particular, for two impurities,
\begin{equation}
E = F\!\left(\bm S_{c,1}\cdot \bm S_{c,2}\right),
\end{equation}
for some function $F$, independent of the coupling strength $J_c$.
\end{lemma}

\begin{proof}
Let $R\in{\rm SO}(3)$ be an arbitrary global spin rotation. Since the host Hamiltonian is SU(2) invariant, there exists a unitary operator $U(R)$ acting on the host Hilbert space such that
\begin{equation}
U(R)\,\hat{\bm S}_\ell\,U(R)^\dagger
=
R\,\hat{\bm S}_\ell
\quad
\text{for all sites }\ell,
\end{equation}
and
\begin{equation}
U(R)\,H_{\rm host}\,U(R)^\dagger
=
H_{\rm host}.
\end{equation}

Treating $\bm S_{c,1}$ and $\bm S_{c,2}$ as external classical vectors, we find
\begin{align}
U(R)\,H(\bm S_{c,1},\bm S_{c,2})\,U(R)^\dagger
&=
H_{\rm host} + J_c\left(\bm S_{c,1}\cdot R\hat{\bm S}_i
+ \bm S_{c,2}\cdot R\hat{\bm S}_j \right) \nonumber\\
&=
H_{\rm host} + J_c\left[ (R^{-1}\bm S_{c,1})\cdot \hat{\bm S}_i
+ (R^{-1}\bm S_{c,2})\cdot \hat{\bm S}_j \right] \nonumber\\
&=
H(R^{-1}\bm S_{c,1},R^{-1}\bm S_{c,2}),
\end{align}
where we used the identity
$\bm a\cdot(R\bm b)=(R^{-1}\bm a)\cdot\bm b$.
Since unitary conjugation leaves the spectrum invariant,
\begin{equation}
{\rm spec}\,H(\bm S_{c,1},\bm S_{c,2})
=
{\rm spec}\,H(R\bm S_{c,1},R\bm S_{c,2}),
\end{equation}
and in particular,
\begin{equation}
E(\bm S_{c,1},\bm S_{c,2})
=
E(R\bm S_{c,1},R\bm S_{c,2}).
\end{equation}


For two spins of fixed magnitude the only independent scalar invariant is their dot product. This proves that the exact impurity--impurity interaction must be of the form
\begin{equation}
E = F(\bm S_{c,1}\cdot \bm S_{c,2})
\end{equation}
independently of the strength of the impurity--host coupling. Importantly, this constraint holds independently of the impurity–host coupling strength and does not rely on perturbation theory.
\end{proof}
\bibliography{bibliography}
\end{document}